\documentclass[a4paper,12pt,oneside]{article}
\usepackage{amsmath,amssymb,mathtools,bm,lipsum}
\usepackage[utf8]{inputenc}
\usepackage{csquotes}
\usepackage[english]{babel}
\usepackage[T1]{fontenc}
\usepackage{amstext}
\usepackage{amsthm}
\usepackage{bbm}
\usepackage{enumerate}
\usepackage{hyperref}
\usepackage[capitalise]{cleveref}
\usepackage{braket}
\usepackage{dsfont}
\usepackage{comment}

\usepackage[verbose=true,letterpaper]{geometry}
\AtBeginDocument{
	\newgeometry{
		textheight=9.5in,
		textwidth=7in,
		top=1in,
		headheight=14.5pt,
		headsep=10pt,
		footskip=20pt
	}
}

\usepackage{fancyhdr}
\fancyheadoffset{0pt}
\usepackage{titlesec}
\titleformat{\section}{\bfseries\scshape\Large}{\thesection}{1em}{}{}
\titleformat*{\subsection}{\scshape\bfseries\large}

\renewenvironment{thebibliography}[1]{
	\begin{oldthebibliography}{#1}
		\footnotesize
		\setlength{\itemsep}{0em}
		\setlength{\parskip}{0em}
	}
	{
	\end{oldthebibliography}
}

\newcommand{\IN}{\mathbb{N}}

\newcommand{\IR}{\mathbb{R}}
\newcommand{\IC}{\mathbb{C}}
\newcommand{\R}{\mathbb{R}}

\newcommand{\N}{\mathbb{N}}

\newcommand{\hs}{\mathfrak{h}}
\newcommand{\HS}{\mathcal{H}}
\newcommand{\FS}{\mathcal{F}}

\newcommand{\hh}{\mathfrak{h}}

\newcommand{\cD}{\mathcal{D}} 
\newcommand{\cQ}{\mathcal{Q}}
\newcommand{\cF}{\mathcal{F}}

\newcommand{\eps}{\varepsilon}
\newcommand{\ph}{\varphi}

\newcommand{\inn}[1]{\langle {#1} \rangle }

\renewcommand{\:}{\colon}
\newcommand{\abs}[1]{\left| #1 \right|}
\newcommand{\nn}[1]{\left\| {#1} \right\| }

\renewcommand{\d}{\,\mathsf{d}}

\newcommand{\dG}{\mathsf{d}\Gamma}

\numberwithin{equation}{section}
\newtheorem{theorem}{Theorem}[section]
\newtheorem{lemma}[theorem]{Lemma}
\newtheorem{proposition}[theorem]{Proposition}

\theoremstyle{definition}

\newtheorem{hyp}{Hypothesis}

\theoremstyle{remark}
\newtheorem{remark}[theorem]{Remark}
\newtheorem{ex}[theorem]{Example}

\crefname{hyp}{Hypothesis}{Hypotheses}
\Crefname{hyp}{Hypothesis}{Hypotheses}
\crefname{lemma}{Lemma}{Lemmas}
\Crefname{lemma}{Lemma}{Lemmas}
\crefname{enumi}{}{}
\Crefname{enumi}{}{}
\creflabelformat{enumi}{#2(#1)#3}
\crefname{equation}{}{}
\Crefname{equation}{}{}

\def\titlename{\scshape On Existence of Ground States in the  Spin Boson Model}
\title{\LARGE\scshape On Existence of Ground States\\in the  Spin Boson Model}
\author{David Hasler\footnote{\texttt{david.hasler@uni-jena.de}} \qquad Benjamin Hinrichs\footnote{\texttt{benjamin.hinrichs@uni-jena.de}} \qquad Oliver Siebert\footnote{\texttt{oliver.siebert@uni-jena.de}}\\ Friedrich-Schiller-University Jena\\\small Department of Mathematics\\[-.5em]\small   Ernst-Abbe-Platz 2\\[-.5em]\small  07743 Jena\\[-.5em] \small Germany}
\newcommand{\shortauthors}{D. Hasler, B. Hinrichs, O. Siebert}

\newcommand{\Rsl}{R_{n}}

\newcommand{\Id}{{\mathds{1}}}

\newcommand{\Hf}{H_{\mathrm f}}
\renewcommand{\i}{\mathsf{i}}

\begin{document}
	
	\maketitle\thispagestyle{empty}
	\begin{abstract}
	\noindent
	We show the existence of ground states in the massless spin boson model without any  infrared regularization. 
Our proof is non-perturbative and relies on a compactness argument. It   works for arbitrary values of the coupling constant under the hypothesis that the second derivative of the ground state energy as a function
of a constant  external magnetic field is  bounded.
	\end{abstract}

\section{Introduction}

The spin boson model describes a quantum mechanical two-level system which is linearly coupled to a  quantized field of bosons.
If the  bosons are relativistic and massless, the model is used as a simplified caricature describing an atom, coarsely approximated by two states,  coupled to 
the quantized electromagnetic field. Although the  model has been extensively investigated, see for example \cite{Spohn.1998,AraiHirokawa.1997,Gerard.2000} and references therein, it is still  an active area of research, cf. \cite{BachBallesterosKoenenbergMenrath.2017,DamMoller.2018b}.

If this system has a ground state, i.e., if the infimum of the spectrum is an eigenvalue, this physically means that  it exhibits binding. Furthermore, ground states are a necessary ingredient to study  scattering theory in quantum field theories. 
In the case of massless bosons or  photons in $\R^d$, we have the dispersion relation $\omega(k)=|k|$.
As a consequence,   the infimum of the spectrum is not  isolated from the rest of the spectrum and establishing existence of a ground  state is non-trivial.
If one imposes a mild infrared regularization of 
the interaction  function $f$, such that the quotient $f/\omega$ is square-integrable, 
e.g., in the case $d=3$  if we have $\delta>-1/2$ such that $f(k) \sim |k|^{\delta }$   for small photon momentum $|k|$, then  existence of ground states has been shown \cite{Spohn.1989,Spohn.1998,AraiHirokawa.1995,AraiHirokawa.1997,BachFroehlichSigal.1998a,BachFroehlichSigal.1998b,Gerard.2000,LorincziMinlosSpohn.2002} 
and its analytic dependence on coupling parameters has been established \cite{GriesemerHasler.2009}.
However, in   models
of physical interest where $d=3$, the coupling function typically has  the behavior $f(k) \sim |k|^{-1/2}$ and $f / \omega$ is
no longer square-integrable. In such a situation the model is infrared-critical in the sense that 
 an infrared problem may  occur and 
a ground state ceases to exist. Such a behavior was most prominently observed  for  translation invariant models in \cite{Frohlich.1973,HaslerHerbst.2008a,DamHinrichs.2021}, see also  references therein. 
Moreover, the absence of ground states was shown for the Nelson model \cite{LorincziMinlosSpohn.2002}
as well as for generalized spin boson models  \cite{AraiHirokawaHiroshima.1999}, provided a nonvanishing expectation condition is 
satisfied.  However, it may also happen in  the infrared-critical case   that the infrared divergences cancel  and a ground state exists.
Heuristically, the reason behind this cancellation is  an underlying   symmetry of the model.
In particular, existence of ground states have been shown for models 
of non-relativistic quantum electrodynamics   \cite{BachFrohlichSigal.1999,GriesemerLiebLoss.2001,HaslerHerbst.2011a,HaslerHerbst.2011b,BachChenFrohlichSigal.2007,HaslerSiebert.2020}.  
 Due to the absence of diagonal entries in the coupling matrix, Herbst and the first author \cite{HaslerHerbst.2010} proved that the spin boson model does actually exhibit a ground state even in the infrared-critical case, see also  \cite{BachBallesterosKoenenbergMenrath.2017} for a recent   alternative proof  providing new insight. 

In this paper, we consider couplings which are more singular  than in  \cite{HaslerHerbst.2010,BachBallesterosKoenenbergMenrath.2017}
and prove the existence of a  ground state in  the spin boson model,  e.g.,   in $d=3$ for any coupling   $f(k)\sim |k|^{\delta}$ for $|k| \to 0$  with $\delta > -1$, provided 
an energy bound  is satisfied. We note, this result is  optimal in the sense that for $\delta = - 1$ the field operator is no longer  bounded in terms of the free field energy. 
In contrast to 
previous results,  our result is non-perturbative and holds for all values of the coupling constant 
as long as the energy inequality holds.
Let us be more precise on the statement.
Denote by $\omega\:\IR^d\to[0,\infty)$ the boson dispersion relation and $f\:\IR^d\to\IR$ the interaction of the quantum field and the two-level system. Then,
 the lower-bounded and self-adjoint Hamilton operator describing the spin boson model acts on the Hilbert space $\IC^2\otimes \FS$, with $\FS$ being the usual Fock space on $L^2(\IR^d)$, and is given as
\begin{equation}\label{def:intro}
H(\omega,f) = \sigma_z \otimes \Id + \Id \otimes \dG(\omega) + \int_{\IR^d} f(k)\sigma_x \otimes (a_k^\dag + a_k)\d k.
\end{equation}
Here, $\dG(\omega)$ denotes the second quantization of the operator of multiplication by $\omega$, moreover $a_k$, $a_k^\dag$ are the distributions describing  annihilation and creation operators,  respectively, and  
$\sigma_x$ and  $\sigma_z$ denote the  Pauli matrices.
A more rigorous definition can be found in \cref{sec:model}.
For the energy inequality we consider $e_n(\mu)=\inf \sigma ( H(\omega_n,f) + \mu(\sigma _x\otimes\Id))$, where the sequence $(\omega_n)_{n\in\IN}$ converges to $\omega$ uniformly and is chosen, such that $f/\omega_n$ is square-integrable. The parameter $\mu\in\IR$ can hereby by interpreted as an external magnetic field.
Explicitly, we assume that the second derivative of $e_n(\mu)$ exists at $\mu=0$ and is bounded as $n\to\infty$, for our result to hold.
This assumption is related to a bound on the magnetic susceptibility of a continuous Ising model. We note that such 
a bound has been shown in our situation for the case $d=3$ and $\delta =-1/2$ \cite{Spohn.1989}. For $d=3$ and  $\delta \in (-1/2,-1)$,  the bound has been shown to hold for the discrete Ising model, cf. \cite{Dyson.1969} and references therein. We believe that a proof of our assumption can be obtained by taking a continuum limit of the discrete Ising model.
As a consequence of our result, non-existence of a ground state for large coupling would imply the divergence of the magnetic susceptibility in the corresponding 
Ising model.

For the proof of our result, we utilize that the 
existence of ground states for the infrared regular situation has been established using a variety of techniques. Hence, if we consider $H(\omega_n,f)$ as above, then a ground state $\psi_n$ exists.
We then prove these ground states lie in a compact set and hence there exists a strongly convergent subsequence $(\psi_{n_k})_{k\in\IN}$. The limit of this sequence will be the ground state of $H(\omega,f)$.

\section{Model and Statement of Results}\label{sec:model}

Throughout this paper we assume $d\in\IN$ and write $\hs=L^2(\IR^d)$ for the state space of a single boson.
Then, let $\cF$ be the bosonic Fock space defined by
\begin{equation}\label{defn:Fockspace}
\cF=\IC\oplus\bigoplus_{n=1}^\infty \FS^{(n)} \qquad\mbox{with}\ \FS^{(n)}= L^2_{\mathrm{sym}}(\IR^{nd}),
\end{equation} 
where we symmetrize over the $n$ $\IR^d$-variables in each component. We write an element $\psi\in \cF$ as $\psi=(\psi^{(n)})_{n \in \IN_0}$ and define the vacuum $\Omega=(1,0,0,\dots)$. 

For a measurable function $\omega:\IR^d\to\IR$,
we define
\begin{align}
\dG(\omega)& = 0  \oplus\bigoplus_{n=1}^{\infty} \omega^{(n)}\qquad\mbox{with}\  \omega^{(n)}(k_1,\ldots,k_n) = \sum_{i=1}^{n}\omega(k_i)
\end{align}
as operators on $\FS$.
Further, for $f\in \hs$, we define the annihilation operator $a(f)$ and creation operator $a^{\dagger}(f)$ using $a(f)\Omega=0$, $a^\dagger(f)\Omega=f$ and for $g\in\FS^{(n)}$
\begin{align}
(a(f)g) (k_1,\ldots,k_{n-1})&={\sqrt{n}}\int \overline{f(k)}g(k,k_1,\ldots,k_{n-1}) \d k\in\FS^{(n-1)},\\
(a^\dagger(f)g)(k_1,\ldots,k_n,k_{n+1})&=\frac{1}{\sqrt{n+1}}\sum_{i=1}^{n+1}f(k_i)g(k_1,\ldots,\widehat{k}_i,\ldots,k_{n+1})\in\FS^{(n+1)},
\end{align}
where $\widehat{k}_i$ means that $k_i$ is omitted from the argument. One can show that these operators can be extended to closed operators on $\cF$ that satisfy $(a(f))^*=a^{\dagger}(f)$. From the creation and annihilation operator, we define the field operator
\begin{equation}\label{defn:fieldoperator}
\varphi(f)=\overline{ a(f)+a^\dagger(f) }.
\end{equation}
The following properties are well-known and can for example be found in \cite{ReedSimon.1975,Arai.2018}.
\begin{lemma}\label{lem:standard}
	Let $\omega,\omega'\:\IR^d\to\IR$
	 and $f\in \hs$. Then
	\begin{enumerate}[(i)]
		\item $\dG(\omega)$ and $\ph(f)$ are self-adjoint.
		\item If $\omega'\ge\omega$, then $\dG(\omega')\ge \dG(\omega)$. Especially, if $\omega\ge 0$, then $\dG(\omega)\ge 0$.
		\item Assume $\omega>0$ almost everywhere and $\omega^{-\frac{1}{2}} f\in \hs$. Then $\varphi(f)$ and $a(f)$ are $\dG(\omega)^{1/2}$-bounded and for $\psi\in\cD(\dG(\omega)^{1/2})$ we have
		\begin{align*}
		\lVert a(f) \psi \lVert&\leq  \lVert \omega^{-\frac{1}{2}}f \lVert  \lVert \dG(\omega)^{\frac{1}{2}}\psi \lVert \quad\mbox{and}\\
		\lVert \varphi(f) \psi \lVert&\leq 2 \lVert (\omega^{-\frac{1}{2}}+1)g \lVert  \lVert (\dG(\omega)+1)^{\frac{1}{2}}\psi \lVert.
		\end{align*}
		In particular, $\varphi(g)$ is infinitesimally $\dG(\omega)$-bounded.
	\end{enumerate}
\end{lemma}
\noindent
Now, let
\begin{equation}\label{eq:space}
\HS=\IC^2\otimes \FS \cong \FS\oplus\FS,
\end{equation}
where the unitary equivalence is implemented by $(v_1,v_2)\otimes \psi \mapsto v_1\psi \oplus v_2\psi$. Furthermore, let $\sigma_x$ and $\sigma_z$ be the usual $2\times2$ Pauli-matrices
\begin{equation}
\sigma_x = \begin{pmatrix}0 & 1 \\ 1 & 0\end{pmatrix} 
\quad\mbox{and}\quad \sigma_z = \begin{pmatrix}1 & 0 \\ 0 & -1\end{pmatrix}.
\end{equation}
For a measurable function $\omega:\IR^d\to\IR$ and  $f\in\hs$, we define the spin boson Hamiltonian
on $\HS$ as
\begin{equation}
H(\omega,f) = \sigma_z\otimes\Id + \Id\otimes \dG(\omega) + \sigma_x\otimes \ph(f) . 
\end{equation}
\begin{lemma}\label{lem:self-adjoint}
	Assume that $\omega:\IR^d\to\IR$   is measurable and  almost everywhere positive, $f\in\hs$,  $\omega^{-1/2}f\in\hs$,  and $\mu\in\IR$.
	Then $H(\omega,f) + \mu \sigma_x \otimes \Id$ defines a lower-bounded self-adjoint operator on the domain $\cD(\Id\otimes\dG(\omega))$. Any core for $\Id\otimes \dG(\omega)$ is a core for $H(\omega,f) + \mu \sigma_x \otimes \Id$.
\end{lemma}  
\begin{proof}
 The operator $K=\sigma_z\otimes\Id+\Id\otimes \dG(\omega)+\mu\sigma_x\otimes\Id$ is self-adjoint 
	as a sum of a  self-adjoint operator with bounded self-adjoint operators and has domain $\cD(K)=\cD(1\otimes\dG(\omega))$.  Moreover, it is bounded from below since $\dG(\omega)$ is non-negative by \cref{lem:standard} and $\sigma_x,\sigma_z$ are bounded. The symmetric operator $\sigma_x\otimes \ph(f)$ is infinitesimally $K$-bounded by \cref{lem:standard}. Hence, the statement follows from the Kato-Rellich theorem (cf. \cite[Theorem X.12]{ReedSimon.1975}).
\end{proof}
\noindent
From now on, we fix a measurable non-negative  function $\omega:\IR^d\to\IR$ and an $f\in\hs$ and work under the following assumptions.
\begin{hyp}\label{hyp1and2}\label{hyp1}\label{hyp2}\ 
	\begin{enumerate}[(i)]
		\item \label{h1:pos} $\omega$ takes positive values almost everywhere,
		\item  \label{h:to infinity} $\omega(k)\xrightarrow{\abs k\to\infty}\infty$.
		\item \label{h:hoelder}
		There exists $\alpha_1 > 0$, such that $\omega$ is locally $\alpha_1$-Hölder continuous.
		\item \label{h:epsilon integrability} There exists $\epsilon > 0$, such that $\omega^{-1/2} f \in L^2(\IR^d)\cap L^{2+\epsilon}(\IR^d)$.
		
		\item \label{h:f hoelder}
		There exists $\alpha_2 > 0$, such that
		\[
		\sup_{|p|\le 1} \int_{\IR^d} \frac{ \abs{f(k+p) - f(k)} }{ \sqrt { \omega(k)} \abs p^{\alpha_2}} \d k < \infty.
		\]	 
		\item \label{h:f L1 integrability}
		We have
		\[
		\sup_{|p|\le 1}\int_{\IR^d} \frac{\abs {f(k)}}{\sqrt{ \omega(k) } \omega(k+p)} \d k < \infty.
		\]	 
	\end{enumerate}
\end{hyp}
\begin{ex}\label{ex:threedim} In $d=3$ dimensions  elementary estimates show that the assumptions of    \cref{hyp1} hold for the choices
\begin{equation} \label{eq:standardform} 
 \omega(k) = |k| \ , \quad f(k) = g \kappa(k) |k|^{\delta}
\end{equation} 
  for any $\delta > -1$, $g \in \R$,  and a cutoff function $\kappa$ of the form  $\kappa(k) = 1_{|k| \leq \Lambda}$, for some  $\Lambda > 0$,  or $\kappa(k) = \exp(-ck^2)$, for some $c > 0$.
The number $g$ will be referred to as the coupling parameter.  
\end{ex} 
\begin{ex} More generally as in \cref{ex:threedim}, we
 consider, 	for  $ k \in \IR^d$, the functions 
	$		\omega(k) = |k|^\alpha,\  f(k)=g\kappa(k)|k|^\beta,
	$ with some $\alpha > 0$, $\beta \in \IR$, $g \in \R$, and  $\kappa$  a cutoff function as in \cref{ex:threedim}. 
	Then \cref{hyp1} holds under the condition $d>\max\{\alpha-2\beta,\frac32\alpha-\beta,\frac12\alpha-2\beta\}$.
\end{ex}
The second assumption we need is a differentiability condition for the infimum of the spectrum with  respect to the constant magnetic field $\mu$. For this, we need to ensure that $\omega$ can be approximated by a family of  functions which are bounded  from below  by some positive constant.
\begin{hyp}\label{hyp3} There exists a decreasing sequence $(\omega_n)_{n\in\IN} $ of nonnegative measurable functions $\omega_n:\R^d\to\IR$ converging uniformly to $\omega$, with the following properties.
	\begin{enumerate}[(i)]
		\item\label{part:omegamprop}\label{eq:Hoelderm} There exists  $\alpha_1 > 0$, such that $\omega_n$ is locally $\alpha_1$-Hölder continuous for all $n\in\IN$.
		\item \label{part:pos}   $\inf_{k \in \IR^d} \omega_n(k) > 0$.
		\item\label{part:diff} \label{eq:hypb0} 
		The function    $e_n(\mu) = \inf\sigma(H(\omega_n,f ) +  \mu \sigma_x  \otimes \Id)$ is  twice differentiable 
		at zero   and  
		\begin{equation} \label{eq:energyderivative} 
		C_\chi := \sup\limits_{n \in \N}(-e_{n}''(0))  <\infty  . 
		\end{equation} 
	\end{enumerate}
\end{hyp}
\begin{remark}
We note that  \cref{part:omegamprop} and \cref{part:pos}  of \cref{hyp3}   are satisfied for the typical choice of a massive photon dispersion relation
  \begin{equation} \label{posdisp}
\omega_n=\sqrt{m_n^2+\omega^2}, 
\end{equation}
or also $\omega_n=\omega+ m_n$, where $(m_n)_{n \in \N}$
  is any sequence of positive numbers decreasing  monotonically to zero.  The constant $m_n$ can be understood to be a photon mass. The result we prove is, however, independent of the specific choice of $\omega_n$.
\end{remark}
\begin{remark}  
	The differentiability  assumption in \cref{hyp3} \cref{part:diff} can be shown to hold by regular analytic perturbation theory
provided   \cref{part:pos} holds, since  \cref{part:pos} implies that the ground state energy is separated from the rest of the spectrum (cf. \cite{AraiHirokawa.1995,AraiHirokawa.1997} or \cref{prop:ground-state}).
However, the uniform bound on the second 
derivative 	is nontrivial to establish. 
	We note that this assumption will be translated into a bound on the resolvent by means of second 
	order perturbation theory, see \cref{lem:1}. 
	In fact, this bound on the resolvent is what  we  need in the proof of the main result, i.e.,   Theorem \ref{maintheorem} holds if one replaces  \cref{part:diff} by 
	the bound in  Lemma \ref{lem:1}.
\end{remark} 
\begin{remark} Let us discuss  \cref{eq:energyderivative}  in  $d=3$ dimensions for the case given in \cref{eq:standardform,posdisp}. 
If $\delta > -1/2$ then   \cref{eq:energyderivative} follows, e.g., from the estimates in  \cite{GriesemerHasler.2009}. In the  case $\delta  = - 1/2$, 
$\cref{eq:standardform}$ has been shown for small values of the coupling constant  in  \cite{Spohn.1989}, using 
that ground state properties of the spin boson Hamiltonian are related to the correlation functions of a continuous  one dimensional 
Ising model with long range interaction, cf. \cite{SpohnDuemcke.1985} and \cite{AizenmanNewman.1986}. 
In particular,  \eqref{eq:energyderivative} translates  to the corresponding Ising model having  finite magnetic suszeptibility. For $\delta \in (-1/2,-1)$ the finiteness
of the magnetic suszeptibility has been shown for the discrete Ising model, \cite{Dyson.1969}.
In fact, for  $\delta \in  [ -1/2,-1)$  the  bound  \cref{eq:energyderivative} does not   
hold  anymore for large values of the coupling constant. This follows from the relation to the Ising model and the  phase transition for  Ising 
models with  coupling decaying quadratically in the distance, cf. \cite{Spohn.1989,AizenmanNewman.1986,ImbrieNewman.1988}.
On the other hand, for small couplings,  using the relation to a continuous   Ising model, it has been shown in \cite{Abdessalam.2011} that $g \mapsto \inf \sigma(H(|\cdot| , g f  ))$ is analytic in a neighborhood 
of zero if $ f , f |\cdot |^{-1/2} \in \hh $, which corresponds to $\delta >  -1$. Note that the infimum of the spectrum can be analytic although there does not 
exist a ground state (cf. \cite{AbdessalamHasler.2012}).
Thus, it is not unreasonable to  suspect that  \eqref{eq:energyderivative}  might in fact hold for $\delta > - 1$ provided the coupling is sufficiently small.
\end{remark}
Our main result now is the following.
\begin{theorem}\label{maintheorem}
	Assume \cref{hyp1,,hyp3} hold. 
	Then $\inf\sigma(H(\omega,f))$ is an eigenvalue of $H(\omega,f)$.
\end{theorem}
\begin{remark}
	This result has been proven for infrared regular models, e.g., under the additional assumption $\inf_{k \in \IR^d} \omega(k) >0$ in \cite{AraiHirokawa.1995} (see \cite{AraiHirokawa.1997} for a generalization of this result) and under the assumption $\omega^{-1}f\in\hs$ in \cite{Gerard.2000}, which  in $d=3$ corresponds to $\delta > -1/2$ in  \cref{eq:standardform}.
Specifically for  $d=3$ the existence has been  shown in situations where  $\omega^{-1}f\notin\hs$ in \cite{HaslerHerbst.2010,BachBallesterosKoenenbergMenrath.2017}.
 The results in these papers include the case \eqref{eq:standardform} 
with $\delta = -1/2$ provided the coupling $g$ is sufficiently  small. The results are perturbative in nature and were obtained using operator theoretic renormalization
and iterated perturbation theory, respectively.  In particular,  \cite{HaslerHerbst.2010} not only shows existence, but also analyticity of the ground state and the 
ground state energy in the coupling constant. 
  Concerning existence,  the result  of \cref{maintheorem} goes beyond. It shows 
existence for any $\delta > - 1$ and arbitrary coupling, as long as the derivative bound  \cref{eq:energyderivative} is  finite. 
\end{remark}
\begin{remark} In \cite{Spohn.1989}  finite temperature KMS states of the spin boson Hamiltonian where investigated for $\int_{\IR^d}  f(k)^2 e^{-\omega(k) |t|} \d k  \cong t^{-2}$  for large $t$.   For $d=3$  this corresponds to $\delta = -1/2$ in \eqref{eq:standardform}.  Using results about the 
one dimensional continuous Ising model,  
it was established that the KMS states 
have a weak  limit as the temperature drops to zero. In particular, it was shown that there 
exists a critical coupling such that the expectation of the number of bosons is 
finite below and infinite at and above the critical coupling strength.  
We note that for the proof of the main theorem we  use 
a similar bound 
on the number of photons, see \cref{coro:photon bounds} \cref{photon bounds 1} 
 (which is in fact weaker than the one in  \cite{Spohn.1989}). 
\end{remark} 
\begin{remark} We note that
our result gives   a physically  explicit bound  on the coupling constant via  \cref{eq:energyderivative}, where the 
left hand side of \cref{eq:energyderivative}  is proportional  to the magnetic 
suszeptibility  of  the corresponding  Ising model.
As a consequence of  \cref{maintheorem}   the absence of  a ground state implies that the magnetic 
suszeptibility must diverge. Given  the existence results in  \cite{HaslerHerbst.2010,BachBallesterosKoenenbergMenrath.2017}, in 
 case $\delta = -1/2$  in \cref{eq:standardform},  the absence of a ground state  for large coupling could provide an alternative method of  proof for    phase transitions 
in  continuous long range Ising models. To the best of our knowledge the absence of a ground state in the spin boson model with $\mu = 0$ 
for $  \delta  \in  (-1, - 1/2]$ and large coupling has not yet been shown. Nevertheless, we  refer the reader to results \cite{Spohn.1989,DamMoller.2018b} where the large coupling 
limit has been investigated.
\end{remark} 
The method of proof we use is based on the proof in \cite{GriesemerLiebLoss.2001}. It was applied to the infrared-critical model of non-relativistic quantum electrodynamics by two of the authors in \cite{HaslerSiebert.2020}.

For the proof of \cref{maintheorem}, we denote by $\psi_n$ the ground state of $H(\omega_n,f)$, which exists due to the assumption $\inf_{k \in \R^d} \omega_n > 0$. 
We then prove, that all of them lie in a compact set $K\subset \IC^2\otimes \FS$. Hence, there exists a subsequence  $(\psi_{n_j})_{j\in\IN}$  converging strongly  to some $\psi\in K$. It then remains to show that $\psi\ne0$ actually is a ground state of $H(\omega,f)$.

The rest of this paper is organized as follows. In \cref{sec:groundstates}, we show some simple properties of the states $\psi_n$ and the corresponding ground state energies. In \cref{sec:bounds}, 
we then derive necessary upper bounds with respect to the photon number to construct the compact set $K$ in \cref{sec:compact}.

\section{Ground State Properties for Massive Photons}\label{sec:groundstates}

In this section, we derive some simple properties of the ground state energy of the infrared regular spin boson Hamiltonian.  
Throughout this section we will assume \cref{hyp1}  and that the sequence $(\omega_n)_{n\in\IN}$ is chosen as in \cref{hyp3} \cref{part:pos}.
We set
\begin{equation}\label{def:simpleops}
	H = H(\omega,f) \quad\mbox{and}\quad H_n = H(\omega_n,f),
\end{equation}
as well as  $E = \inf \sigma(H)$ and  $E_n = \inf \sigma(H_n)$ for all $n \in \N$. 
\begin{lemma}\label{lem:energylimit} We have
	\begin{enumerate}[(i)]
		\item\label{part:operatormonotonicity} $H \le H_{n'}\le H_{n}$ for $n \leq n' $,
		\item\label{part:energylimit} $\lim\limits_{n\to\infty} E_n = E$.
	\end{enumerate}
\end{lemma}
\begin{proof}
	\cref{part:operatormonotonicity} follows from the monotonicity of $(\omega_n)$ and \cref{lem:standard}.
	We set $N=\Id \otimes \dG(1)$. Then, due to the uniform convergence of $(\omega_n)$, there is a sequence $(C_n)\subset \IR^+$ satisfying $C_n\xrightarrow{n\to\infty} 0$ and $\omega_n\le\omega+C_n$. Hence,
	\[ \dG(\omega_n) \le \dG(\omega) + C_n\dG(1), \qquad\mbox{which implies}\qquad H_n\le H + C_n N. \]
	On the other hand  let $\eps>0$ and fix $\varphi_\eps\in\cD(N)\cap \cD(H_0)$ with $\|\varphi_\eps\|=1$, such that
	\[ \braket{\varphi_\eps,H\varphi_\eps}\le E + \eps.\]
	This is possible,  since $\cD(N)\cap \cD(H_0)$ is a core for $\Id\otimes \dG(\omega)$ and hence for $H$, by \cref{lem:self-adjoint}. Together with \cref{part:operatormonotonicity}, we obtain
	\[ E\le E_n\le \braket{\varphi_\eps,H_n\varphi_\eps}\le \braket{\varphi_\eps,H\varphi_\eps}+ C_n \braket{\varphi_\eps,N\varphi_\eps}\le E+ \eps + C_n\braket{\varphi_\eps,N\varphi_\eps} \xrightarrow{n\to\infty} E + \eps. \]
	Now \cref{part:energylimit} follows in the limit $\eps\to 0$.
\end{proof}
\noindent
As mentioned above the bound \cref{hyp3}  \cref{part:pos}  implies the existence of a ground state, which is
the content of the following lemma. 
\begin{proposition}\label{prop:ground-state}
	For all $n\in\IN$, $E_n$ is a simple eigenvalue of $H_n$.\\ Further, $\left[E_n,E_n+\inf_{k\in\IR^d}\omega_n(k)\right)\cap\sigma_{\mathrm{ess}}(H_n)=\emptyset$.
\end{proposition}
\begin{proof}
The existence has been shown in \cite{AraiHirokawa.1995} and the uniqueness for arbitrary
couplings has been shown in \cite{HaslerHerbst.2010}, see also \cite{Frohlich.1973}.
\end{proof} 
\noindent
Let $\psi_n$ be a  normalized eigenvector of $H_n$ to the eigenvalue $E_n$.
A main ingredient of our proof then is the following \lcnamecref{prop:minimizingsequ}.
\begin{proposition}\label{prop:minimizingsequ}
	The sequence $(\psi_n)_{n\in\IN}$ is minimizing for $H$, i.e.,
	\[ 0\le \braket{\psi_{n},(H-E)\psi_{n}}\xrightarrow{n\to \infty}0. \]
\end{proposition}
\begin{proof}
	We use \cref{lem:energylimit} and find
	\begin{equation*}
		0\le \braket{\psi_{n},(H-E)\psi_{n}} \le \braket{\psi_{n},(H_{n}-E)\psi_{n}} = E_{n}-E \to 0.
		\qedhere
	\end{equation*}
\end{proof}

\section{Infrared Bounds}\label{sec:bounds}

In this section we derive essential bounds on the ground states $\psi_n$, which are uniform in $n\in\IN$. Throughout this section we will assume that  \cref{hyp1,,hyp3} hold.
We recall the following  definition in  \cref{hyp3}
\begin{equation}
	 e_n(\mu)=\inf\sigma( H(\omega_n,f) + \mu \sigma_x  \otimes \Id )\quad\mbox{for}\ n\in\IN\ \mbox{and}\ \mu\in\IR.
\end{equation}
Note  that by the definitions in \cref{def:simpleops}, we have $E_n = e_n(0) $. 
The next   \lcnamecref{lem:sym} is a simple symmetry argument.
\begin{lemma} \label{lem:sym} We have $e_{n}(\mu) = e_{n}(-\mu)$ for all $n\in\IN$. 
\end{lemma} 
\begin{proof}
	We define the unitary operator $U=e^{  \i \frac{ \pi }{2} \sigma_z  }\otimes(-1)^{\dG(1)}$. It easily follows from the definitions that $
	e^{  \i \frac{ \pi }{2} \sigma_z  } \sigma_x e^{-  \i \frac{ \pi }{2} \sigma_z  }  = - \sigma_x$ and $(-1)^{\dG(1)}\ph(f)(-1)^{\dG(1)}=\ph(-f)$. Now, using that $\dG$-operators commute, we obtain
	$U  ( H(\omega_n,f) + \mu \sigma_x \otimes \Id ) U^* = H(\omega_n,f) -\mu \sigma_x \otimes \Id$, which proves the claim.
\end{proof} 
\noindent
Now, for $k\in\IR^d$, we define the  pointwise  annihilation operator $a_k$ acting on $\psi^{(\ell+1)}\in\FS^{(\ell+1)}$ by 
\begin{equation}\label{def:pointwise}
	 (a_k\psi^{(\ell+1)})(k_1,\ldots,k_\ell)=\sqrt{\ell+1}\psi^{(\ell+1)}(k,k_1,\ldots,k_\ell) .
\end{equation}
Note that by the Fubini-Tonelli theorem  $ (a_k\psi^{(\ell+1)}) \in \FS^{(\ell)}$ for almost every $k$.
Further, for $n\in\IN$, we define the operator
\begin{equation}\label{eq:Rsl}
	 \Rsl(k)=(H_{n}-E_{n}+\omega_n(k))^{-1}   \quad\mbox{for}\  k \in \R^d , 
\end{equation}
which is bounded by   \cref{hyp3}, and 
the spectral theorem directly yields
\begin{equation} \label{4.4} 
	\|\Rsl(k)\|\le \frac{1}{\omega_n(k)}.
\end{equation} 
The next statement is well-known and can be found under the name pull-through formula throughout the literature, cf. \cite{BachFroehlichSigal.1998b,Gerard.2000}. 
In the statement we write $\psi_n=(\psi_{n,1},\psi_{n,2})$ in the sense of \cref{eq:space} and denote
\begin{equation}
	a_k\psi_n  = (a_k\psi_{n,1},a_k\psi_{n,2}) \quad\mbox{and}\quad \sigma_x\psi_n = (\sigma_x\otimes\Id)\psi_n = (\psi_{n,2},\psi_{n,1}) .
\end{equation}
\begin{lemma}\label{lem:pullthrough}
	Let $n\in\IN$. Then, for almost every $k\in\IR^d$, the vector $a_k\psi_{n} \in\HS$  and
	\[ a_k \psi_{n} = -f(k)\Rsl(k)\sigma_x\psi_{n}.\]
\end{lemma}
\noindent
The infrared bounds we want to obtain in this section are bounds on $\Rsl(k)\sigma_x\psi_n$. To that end, we start by translating \cref{hyp3} into a resolvent bound.
\begin{lemma} \label{lem:1}  For all $n\in\IN$, we have $\braket{\sigma_x \psi_n,\psi_n}=0 $ and
	\[
	0 \leq \inn{ \sigma_x \psi_n , ( H_{n} - E_{n} )^{-1}   \sigma_x  \psi_n } =    - \frac{1}{2} e_{n}''(0).
	\]
\end{lemma}
\begin{proof} The proof uses analytic perturbation theory, for details see \cite{Kato.1980,ReedSimon.1978}.
	  The operator valued function $\eta \mapsto H_n(\eta) := H_n + \eta\sigma_x\otimes \Id$ 
	 defines an analytic family of type (A) for $\eta\in\IC$ (cf. \cite[Theorem 2.6]{Kato.1980}). 	By  \cref{prop:ground-state}, we
	 know that  $e_{n}(0)$ is a non-degenerate eigenvalue of $H_n(0)$ isolated from the essential spectrum.
	First order perturbation theory now yields $e_n'(0) = \braket{\psi_n,\sigma_x\psi_n}$. Hence, \cref{lem:sym} implies $\braket{\psi_n,\sigma_x\psi_n}=0$. Then, by second order perturbation theory, we obtain the second equality.
\end{proof}
\noindent
This gives us the required infrared bound.
\begin{lemma} \label{lem:estimate} 
We have
$\displaystyle \| \Rsl(k) \sigma_x  \psi_m \| \leq \sqrt{ \frac{-e_n''(0)}{\omega_n(k)} } 
$ for all $n\in\IN$.
\end{lemma}
\begin{proof}
By  the product inequality, we have
\begin{equation} \label{eq:prod}
\| \Rsl(k) \sigma_x \psi_n \| \leq  \| R_n(k)
( H_n - E_n )^{1/2} \| \| ( H_n - E_n )^{-1/2}  \sigma_x \psi_n \| .
\end{equation}
By \cref{lem:1}, the second factor on the right hand side can be estimated using
\[
\| (H_n - E_n)^{-1/2}  \sigma_x  \psi_n \| \leq  \sqrt{ -  e_n''(0) } \; .
\]
It remains to estimate the first factor in \cref{eq:prod}. Using $\nn{R_n(k)^{1/2} (H_n - E_n)^{1/2}} \leq 1$ we find with \eqref{4.4} 
\[
\nn{ R_n(k)  (H_n - E_n)^{1/2} } \leq   \nn{ R_n(k)^{1/2} } \leq \frac{1}{\sqrt{ \omega_n(k)}}.
\qedhere \]
\end{proof}
\noindent
We combine this result with the pull-through formula.
\begin{lemma}\label{coro:photon bounds} 
Let  $B_1 = \{ x \in \IR^d : | x | \leq 1 \}$. 
	\begin{enumerate}[(i)]
		\item \label{photon bounds 1} For all $n \in \IN$ and almost all  $k \in \IR^d$, we have 
			$\displaystyle \|a_k\psi_{n}\| \le \frac{ \abs{f(k)}}{\sqrt{\omega(k)}} C_\chi^{1/2}$.
		\item\label{photon bounds 2} There exist  an $\alpha > 0$ and a measurable function $h \: B_1 \times \IR^d \rightarrow [0,\infty)$ with  \[
		 \sup_{p \in B_1} \nn{ h(p,\cdot)}_1 < \infty,
		 \]  such that for  all  $n \in \N$ and almost all   $p \in  B_1$ and $k \in \R^d$  
		 \[ \|a_{k+p}\psi_{n}-a_k\psi_{n}\| \le   \abs p^\alpha  h(p,k). \]
	\end{enumerate}
\end{lemma}
\begin{proof}
	\cref{photon bounds 1} follows directly from \cref{lem:pullthrough,lem:estimate}, and from the monotonicity of $(\omega_n)_{n \in \IN}$.
	
	Let $\alpha_1$ be the minimum of the values from \cref{hyp2} \cref{h:hoelder} and \cref{hyp3} \cref{part:omegamprop} and let $\alpha_2$ be as in \cref{hyp2} \cref{h:f hoelder}.
	Then, we set $\alpha = \min\{\alpha_1,\alpha_2\}$ and
	\[\tilde{h}(p,k) = \max\left\{\frac{ \abs{f(k+p) - f(k)} }{ \abs p^{\alpha}\sqrt{ \omega (k) }  },\frac{ \abs{f(k+p)}  }{ \omega(k) \sqrt{ \omega(k+p)} }\right\}.\]
	Then, by \cref{hyp2}, $\tilde h$ satisfies the above statements on $h$.
	Further, using the  resolvent identity  and \cref{lem:pullthrough}, we obtain
	\begin{align}
		a_{k+p}\psi_{n}-a_k\psi_{n} &=   f(k)\Rsl(k)\sigma_x\psi_{n} - f(k+p)\Rsl(k+p)\sigma_x\psi_{n}  \nonumber  \\
		&= (f(k) -f(k+p))\Rsl(k)\sigma_x\psi_{n} + f(k+p)(\Rsl(k)  -\Rsl(k+p) )\sigma_x\psi_{n} \nonumber \\
		&=  (f(k) -f(k+p))\Rsl(k)\sigma_x\psi_{n} \label{term 1}  \\&\qquad+ f(k+p)\Rsl(k)(\omega_n(k+p)-\omega_n(k))\Rsl(k+p)\sigma_x\psi_{n} \label{term 2} .
	\end{align}
By \cref{lem:estimate,hyp3}, we find
\[
\abs{\eqref{term 1}} \leq C_\chi^{1/2} \frac{ \abs{f(k+p) - f(k)} }{ \sqrt{ \omega (k) }} \leq C_\chi^{1/2} \abs p^{\alpha} \tilde{h}(p,k).
\]
Further, the local $\alpha_1$-Hölder continuity of $\omega_n$ yields there is $C>0$, such that
\[
\abs{\eqref{term 2}} \leq C \abs p ^{\alpha} \tilde{h}(p,k).
\]
This proves the statement for the function $h = (C_\chi^{1/2} + C ) \tilde{h}$. 
\end{proof}
\noindent
We use the above infrared bounds to derive an upper bound on the photon number operator and the free field energy
\begin{equation}\label{def:numberop}
	N = \Id \otimes \dG(1) \quad\mbox{and}\quad \Hf = \Id\otimes \dG(\omega)
\end{equation}
acting on the ground states $\psi_n$.
The proof uses the following well-known representation of the quadratic form associated with second quantization operators in terms of pointwise annihilation operators.
\begin{lemma}\label{lem:integralrep}
	Assume $A\:\IR^d \rightarrow  [0,\infty)$ is measurable and $\psi\in\FS$.  Then the map $k \mapsto \|A(k)^{1/2}a_k\psi\| $ is in  $L^2(\IR^d)$ if and only if $\psi\in\cD(\dG(A)^{1/2})$. Further, for any $\phi_1,\phi_2\in \cD(\dG(A)^{1/2})$ we have
	\[  \braket{\dG(A)^{1/2}\phi_1,\dG(A)^{1/2}\phi_2} =  \int_{\IR^d} A(k) \braket{a_k\phi_1,a_k\phi_2} \d k.  \]
\end{lemma}
\begin{proof}
	The statement is standard in the literature,  see for example \cite{ReedSimon.1975}.
\end{proof}
\noindent
The next lemma will provide  a photon number bound.
\begin{lemma}\label{lem:numberbound} For all $n \in \IN$ 
	we have $\psi_n\in\cD(N^{1/2})\cap\cD(\Hf)$ and the inequalities $\braket{N^{1/2}\psi_n,N^{1/2}\psi_n}\le C_\chi \|\omega^{-1/2}f\|^2$ and $\braket{\psi_n,\Hf\psi_n}\le  C_\chi \|f\|^2$.
\end{lemma}
\begin{proof}
	The property $\psi_n\in\cD(\Hf)$ was proven in \cref{lem:self-adjoint}. The remaining statements follow from combining the upper bound in \cref{coro:photon bounds} \cref{photon bounds 1} and \cref{lem:integralrep}.
\end{proof}

\section{The Compactness Argument}\label{sec:compact}

In this section, we construct a compact set $K\subset \HS$, such that $(\psi_n)_{n\in\IN}\subset K$. We then use the compactness of $K$ to prove \cref{maintheorem}. Throughout this section, we assume that  \cref{hyp1,,hyp3} hold.

Let us begin with the definition  of $K$. To that end, assume $y_i$ for $i=1,\ldots,\ell$ is the position operator acting on $\psi^{(\ell)}\in\FS^{(\ell)}$ as
\begin{equation}
	\widehat{y_i\psi^{(\ell)}}(x_1,\ldots,x_\ell) = x_i\widehat{\psi^{(\ell)}}(x_1,\ldots,x_\ell),
\end{equation}
where $\widehat{\cdot}$ denotes the Fourier transform. 
For $\delta>0$, we now define a closed quadratic form $q_\delta$ acting on $\phi=(\phi_1,\phi_2)\in\cQ(q_\delta)\subset \HS$ with natural domain as
\begin{equation}\label{def:form}
	q_\delta(\phi)  =  \braket{N^{1/2}\phi,N^{1/2}\phi}  + \sum_{\substack{\ell\in\IN\\s\in\{1,2\}}} \frac{1}{\ell^2}\sum_{i=1}^\ell \Braket{ \phi_s^{(\ell)}, \lvert y_i\lvert^{\delta}\phi_s^{(\ell)}  } + \braket{\Hf^{1/2}\phi,\Hf^{1/2}\phi},
\end{equation}
where $N$ and $\Hf$ are defined as in \cref{def:numberop}. Now define
\begin{equation}
	K_{\delta,C} := \{\phi\in\cQ(q_\delta): \|\phi\|\le 1,q_\delta(\phi)\le C \} \qquad\mbox{for}\ C>0.
\end{equation}
\begin{lemma}\label{lem:compact}
	For all $\delta,C>0$ the set $K_{\delta,C}\subset \HS$ is compact.
\end{lemma}
\begin{proof}
	By \cref{lem:standard}, $q_\delta$ is nonnegative. Hence, there exists a self-adjoint nonnegative operator $T$ associated to $q_\delta$. By the general characterization of operators with compact resolvent (cf. \cite[Theorem XIII.64]{ReedSimon.1978}), $K_C$ is compact iff $T$ has compact resolvent iff the $i$-th eigenvalues of $T$ obtained by the min-max principle $\mu_i(T)$  tend to infinity, i.e., $\lim\limits_{i\to\infty}\mu_i(T)=\infty$. 
	
	To that end, we observe $T$ preserves the $\ell$ photon sectors $\IC^2\otimes \FS^{(\ell)}$ and denote $T_\ell = T\upharpoonright\IC^2\otimes \FS^{(n)}$. Now, since $(\omega+1)^{(\ell)}(K)\to\infty$ as $K\to\infty$  by \cref{hyp2} \cref{h:to infinity}, we can apply Rellich's criterion (cf. \cite[Theorem XIII.65]{ReedSimon.1978}) and hence $T_\ell$ has compact resolvent for all $\ell\in\IN_0$. As argued above, we have $\lim\limits_{i\to\infty}\mu_i(T_\ell)=\infty$. Further, since $T_\ell\ge \ell$, we have $\mu_i(T_\ell)\ge \ell$ and therefore $\lim\limits_{i\to\infty}\mu_i(T)=\infty$.
\end{proof}
\noindent
We now need to prove the following \lcnamecref{prop:eigenfunct-compact}, where $\psi_n$ are the normalized ground states of $H_n$ as defined in \cref{sec:groundstates}.
\begin{proposition}\label{prop:eigenfunct-compact}
	There are $\delta,C>0$, such that $\psi_n\in K_{\delta,C}$ for all $n\in\IN$.
\end{proposition}
\noindent
For the proof the following \lcnamecref{lem:fouriertransformbound} is essential. Hereby, for $n\in\IN$, $s\in\{1,2\}$ and $y,k\in\IR^d$, we introduce the notation
\begin{equation}
	\begin{aligned}
		&\widehat{\psi_{n,s}^{(\ell)}}(y) \: (y_1,\ldots,y_{\ell-1}) \mapsto \psi_{n,s}^{(\ell)}(y,y_1,\ldots,y_{\ell-1}),\\
		&{\psi_{n,s}^{(\ell)}}(k) \: (k_1,\ldots,k_{\ell-1}) \mapsto \psi_{n,s}^{(\ell)}(k,k_1,\ldots,k_{\ell-1}).
	\end{aligned}
\end{equation}
Due to the Fubini-Tonelli theorem, we have $\widehat{\psi_{n,s}^{(\ell)}}(y),{\psi_{n,s}^{(\ell)}}(k)\in L^2(\IR^{(\ell-1)d})$ for almost every $k,y\in\IR^d$. Further, comparing with the definition \cref{def:pointwise}, we observe
\begin{equation} \label{eq:relakpoint}   {\psi_{n,s}^{(\ell)}}(k) = \frac{1}{\sqrt{\ell+1}}(a_k \psi_{n,s})^{(\ell)}.\end{equation} 
\begin{lemma}\label{lem:fouriertransformbound}
	There exist $\delta > 0$ and  $C>0$, such that for all $p\in\IR^d$ and $n,\ell\in\IN$, $s\in\{1,2\}$
	\begin{equation} \label{eq:estcomp1} 
	\int_{\IR^d} |1-e^{-ipy}|^2\left\|\widehat{\psi_{n,s}^{(\ell)}}(y)\right\|^2_{L^2(\IR^{(\ell-1)d})} \d y \le \frac C{\ell+1} \min\left\{1,\abs p^\delta\right\}.
	\end{equation} 
We note that $\delta$ can be chosen as $\delta = \dfrac{\epsilon \alpha}{1+\epsilon}$, where $\alpha>0$ and $\epsilon>0$ are as in \cref{coro:photon bounds} \cref{photon bounds 2} and \cref{hyp2} \cref{h:epsilon integrability}, respectively.
\end{lemma}
\begin{proof}
	That the left hand side of  \cref{eq:estcomp1} is bounded 
by a constant $C$, uniformly in $p$,  follows easily due to the Fock space definition, since the Fourier transform preserves the $L^2$-norm.
	Now lets consider $|p| \leq 1$. Note  that
	\begin{align*}
		\int_{\IR^d} |1-e^{-ip y}|^2\left\|\widehat{\psi_{n,s}^{(\ell)}}(y)\right\|^2_{L^2(\IR^{(\ell-1)d})} \d y
		&= \int_{\IR^d} \nn{ \psi_{n,s}^{(\ell)}(k+p) - \psi_{n,s}^{(\ell)}(k)}^2 \d k\\
		&= \frac{1}{\ell+1}\int_{\IR^d} \nn{ (a_{k+p}\psi_{n,s})^{(\ell)} - (a_k\psi_{n,s})^{(\ell)}}^2 \d k , 
	\end{align*}
where we used \cref{eq:relakpoint}.
	Let $\theta \in (0,1)$. By \cref{coro:photon bounds}, we have some $C>0$ such that
	\[\nn{ (a_{k+p}\psi_{n,s})^{(\ell)} - (a_k\psi_{n,s})^{(\ell)}} \le C|p|^{\theta\alpha}h(p,k)^\theta\left(  \frac{\abs{f(k)}}{\sqrt{\omega(k)}}  \right) ^{1-\theta}. \]
	For $r,s > 1$ with $\frac1r+\frac1s = 1$, we now use Young's inequality $b c \leq b^r/r + c^s/s$ to obtain a constant $C_{r,s}>0$ with
	\begin{equation}
		\nn{ (a_{k+p}\psi_{n,s})^{(\ell)} - (a_k\psi_{n,s})^{(\ell)}}^2
		\le C_{r,s} \abs p^{2\theta \alpha } \left(  h(p,k)^{2\theta r} +   \left(  \frac{\abs{f(k)}}{\sqrt{\omega(k)}}  \right) ^{2(1-\theta)s }  \right).\label{eq:integrand summands}
	\end{equation}
	Set $r = \frac{1}{2\theta}$. Then, the first summand in \eqref{eq:integrand summands} is integrable in $k$ due to  \cref{coro:photon bounds}. Further, the exponent of the second summand equals
	\[
	2(1-\theta)s  = 2(1-\theta) \left(1-\frac{1}{r}\right)^{-1} =   \frac{2(1-\theta)}{1-2\theta}.
	\]
	Hence, we can choose $\theta > 0$ such that $ \dfrac{2(1-\theta)}{1-2\theta} = 2 + \epsilon$. By \cref{hyp2} \cref{h:epsilon integrability}, it follows that \cref{eq:integrand summands} is integrable in $k$ and the proof is complete.
\end{proof}
\noindent
From here, we can prove an upper bound for the Fourier term in \cref{def:form}.
\begin{lemma}\label{prop:localizationbound}
	Let $\delta>0$ be as in \cref{lem:fouriertransformbound}. Then there exists $C>0$ such that for all $n,\ell\in\IN$ and $s\in\{1,2\}$
	\[ \int_{\IR^{d\cdot\ell}}\sum_{i=1}^{\ell}  |x_i|^{\delta/2}\left|\widehat{\psi_{n,s}^{(\ell)}}(x_1,\ldots,x_\ell)\right|^2\d(x_1,\ldots,x_\ell) \le C.  \]
\end{lemma}
\begin{proof}
From \cref{lem:fouriertransformbound}, we know that there exists a finite constant $C$ such
that
\[
\int_{\IR^d}\int_{\IR^{d}} \frac{|1 - e^{-ipy} |^2 \| \widehat{\psi_{n,s}^{(\ell)}}(y)
\|^2}{|p|^{\delta/2} } \d y \frac{\d p}{|p|^d} \leq \frac C{\ell+1}   \; .
\]
After interchanging the order of integration  and a change of
integration variables  $q = |y| p $, we find
\[
 \frac{C}{\ell+1}  \geq \int_{\IR^{d}} \| \widehat{\psi_{n,s}^{(\ell)}}(y) \|^2 \int_{\IR^d} \frac{ | 1 -
e^{- i p y} |^2}{|p|^{\delta/2} } \frac{\d p}{|p|^d} \d y = \int_{\IR^{d}} \|
\widehat{\psi_{n,s}^{(\ell)}}(y) \|^2 |y|^{\delta/2} \underbrace{ \int_{\IR^d} \frac{
| 1 - e^{- i q y/|y|} |^2}{|q|^{\delta/2} } \frac{ \d q}{|q|^d}  }_{=: \ c}  \d y 
\; ,
\]
where $c$ is nonzero and does not depend on $y$.	
\end{proof}
\noindent
We can now conclude.
\begin{proof}[\textbf{Proof of \cref{prop:eigenfunct-compact}}]
	Combine \cref{lem:numberbound,prop:localizationbound}.
\end{proof}
\begin{proof}[\textbf{Proof of \cref{maintheorem}.}]
By \cref{lem:compact,prop:eigenfunct-compact}, we know
there exists a subsequence $(\psi_{n_k})_{k\in\IN}$, which converges to a normalized vector
$\psi_\infty$. By the lower semicontinuity of non-negative quadratic forms, we
see from \cref{prop:minimizingsequ}  that
\[ \inn{  \psi_\infty , ( H  - E ) \psi_\infty }  \leq \liminf_{k \to \infty} \inn{  \psi_{n_k} , ( H - E)
\psi_{n_k} } = 0. \qedhere
\]
\end{proof}

\section*{Acknowledgements} 

D.H. wants to thank Ira Herbst for valuable discussions on the subject. 

\bibliographystyle{halpha-abbrv}
\bibliography{lit}

\end{document}